\documentclass[12 pt, letterpaper]{article}  

\usepackage{graphics} 
\usepackage{epsfig} 
\usepackage{mathptmx} 
\usepackage{times} 
\usepackage{amsmath} 
\usepackage{amssymb}  
\usepackage{amsthm}
\usepackage{mathtools}
\usepackage{booktabs}

\usepackage{calrsfs}
\DeclareMathAlphabet{\pazocal}{OMS}{zplm}{m}{n}

\usepackage{cite}
\usepackage{color}
\usepackage{subcaption}
\usepackage{booktabs}

\usepackage{algorithm,algcompatible,amsmath}

\DeclareMathOperator*{\argmin}{arg\,min}
\algnewcommand\INPUT{\item[\textbf{Input:}]}%
\algnewcommand\OUTPUT{\item[\textbf{Output:}]}%

\usepackage{tikz}
\usetikzlibrary{arrows.meta,positioning}
\newcommand{\Cross}{\mathbin{\tikz [x=1.4ex,y=1.4ex,line width=.1ex] \draw (0,0) -- (1,1) (0,1) -- (1,0);}}%

\usepackage{xcolor}

\usepackage{textcomp}
\usepackage[margin=1in]{geometry}
\usepackage{setspace}

\newtheorem{theorem}{Theorem}[section]

\newtheorem{prop}{Proposition}[section]

\newtheorem{definition}{Definition}[section]

\begin{document}

\onehalfspacing
\title{Online Computation of Terminal Ingredients in Distributed Model Predictive Control for Reference Tracking *}
\author{Ahmed Aboudonia$^{1}$, Goran Banjac$^{1}$, Annika Eichler$^{2}$, and John Lygeros$^{1}$}
\date{}

\maketitle
\thispagestyle{empty}
\pagestyle{empty}

\begin{abstract}

A distributed model predictive control scheme is developed for tracking piecewise constant references where the terminal set is reconfigured online, whereas the terminal controller is computed offline. Unlike many standard existing schemes, this scheme yields large feasible regions without performing offline centralized computations. Although the resulting optimal control problem (OCP) is a semidefinite program (SDP), an SDP scalability method based on diagonal dominance is used to approximate the derived SDP by a second-order cone program. The OCPs of the proposed scheme and its approximation are amenable to distributed optimization. Both schemes are evaluated using a power network example and compared to a scheme where the terminal controller is reconfigured online as well. It is found that fixing the terminal controller results in better performance, noticeable reduction in computational cost and similar feasible region compared to the case in which this controller is reconfigured online.
\end{abstract}

\let\thefootnote\relax\footnote{* This work is supported by the European Research Council under the ERC Advanced Grant agreement no. 787845 (OCAL) and the Swiss National Science Foundation under NCCR Automation.}
\let\thefootnote\relax\footnote{$^{1}$Ahmed Aboudonia, Goran Banjac and John Lygeros are with the Automatic Control Laboratory, Department of Electrical Engineering and Information Technology, ETH Zurich, 8092 Zurich, Switzerland {\tt\small $\{$ahmedab,gbanjac,lygeros$\}$@control.ee.ethz.ch}}
\let\thefootnote\relax\footnote{$^{2}$Annika Eichler is with the Deutsches Elektronen-Synchroton DESY, 22607 Hamburg, Germany {\tt\small annika.eichler@desy.de}}

\section{INTRODUCTION}

Control of interconnected systems is an active area of research due to its wide variety of applications \cite{maestre2014distributed}. Various control techniques have been developed to control such systems \cite{scattolini2009architectures}. Among these techniques is Model Predictive Control (MPC), that aims to optimize performance while ensuring stability and constraint satisfaction \cite{rawlings2017model}. Various efforts have been devoted to developing distributed MPC schemes for interconnected systems \cite{christofides2013distributed}. In these schemes, the system is decomposed into several smaller coupled subsystems, each of which has a local controller which can share information with a set of other local controllers.

To ensure stability and constraint satisfaction at all times, many MPC schemes, including distributed schemes, either use sufficiently long horizons or are equipped with terminal ingredients computed offline (e.g. see \cite{kouvaritakis2016model,ferramosca2013cooperative,conte2016distributed}).
While the former may result in high computational cost \cite{kouvaritakis2016model}, the latter may require offline centralized computations \cite{ferramosca2013cooperative} or lead to small feasible regions \cite{conte2016distributed} and hence, deteriorate the closed-loop performance. To circumvent these challenge, several studies have considered computing terminal ingredients for distributed MPC online \cite{trodden2017distributed,lucia2015contract,aboudonia2020distributed}. A distributed MPC scheme with reconfigurable terminal ingredients for piecewise constant reference tracking was proposed in \cite{aboudonia2021distributed} where the terminal cost, controller and set are updated online taking into consideration the current state of the system. The numerical study in \cite{aboudonia2021distributed} shows that this scheme can outperform standard distributed MPC schemes.


In this paper, we develop a distributed MPC scheme for piecewise constant reference tracking where the terminal set and cost are reconfigured online. Unlike \cite{aboudonia2021distributed}, however, the terminal controller is computed offline and no longer considered a decision variable in the online optimal control problem (OCP). The resulting OCP has fewer decision variables in this case, but also fewer constraints; in particular, the constraints required to ensure the stability of the terminal dynamics are no longer added in the OCP and the terminal set is no longer required to be centered around the reference trajectory. Furthermore, constraint satisfaction inside the terminal set is ensured by a set of linear inequalities instead of linear matrix inequalities (LMIs). The flexibility of the terminal set center maintains a large MPC feasible region. Moreover, the reduction in the number of decision variables and constraints (LMIs in particular) results in a remarkable reduction in the online computational cost.
Hence, when solving the OCP using distributed optimization techniques such as the alternating direction method of multipliers (ADMM), more iterations can be performed within the available sampling time, leading to better convergence to the optimal solution. 
We evaluate the efficacy of the proposed approach using a power network case study and observe that both computation and suboptimality with respect to centralized solutions are reduced compared to \cite{aboudonia2021distributed}, while the set of initial conditions for which the OCP is feasible remains similar.

In Section~\ref{sec:problem}, we formulate the tracking distributed MPC problem. In Section~\ref{sec:mpc}, we show how to convert the resulting OCP into a semidefinite program (SDP) and how to approximate the SDP by a second order cone program (SOCP). In Section~\ref{sec:implementation}, we show how the resulting convex OCPs can be solved by ADMM. We explore the power network example in Section~\ref{ref:simulations} followed by some concluding remarks in Section~\ref{conclusion}.

\section{Problem Formulation}\label{sec:problem}

We consider interconnected systems which can be decomposed into $M$ coupled subsystems. Two subsystems are neighbors if the state of one appears in the dynamics of the other, or their states appear jointly in one of the state constraints defined below in \eqref{sec2_cons}. The set of neighbors of the $i^{\text{th}}$ subsystem is denoted as $\pazocal{N}_i$; by convention we assume that $i \in \pazocal{N}_i$ for all $i =\{1, \ldots, M\}$. We let $x_i \in \mathbb{R}^{n_i}$ and $u_i \in \mathbb{R}^{m_i}$ denote the states and inputs of the $i^{\text{th}}$ subsystem. The dynamics of the $i^{\text{th}}$ subsystem is given by
\begin{equation}
\label{sec2_dyn}
x_i(t+1) = A_ix_{N_i}(t)+B_i u_i(t),
\end{equation}
where $A_i \in \mathbb{R}^{n_i \times n_{N_i}}$, $B_i \in \mathbb{R}^{n_i \times m_i}$ and $x_{N_i} \in \mathbb{R}^{n_{N_i}}$ is a concatenated vector including the states of the subsystems in the set $\pazocal{N}_i$. For $j \in \pazocal{N}_i$, the local states of the $j^{\text{th}}$ subsystem can be extracted from the vector $x_{N_i}$ using a binary matrix $W_{ij} \in \{0,1\}^{ n_j \times n_{N_i}}$ where $x_j = W_{ij} x_{N_i}$. The $i^{\text{th}}$ subsystem is subject to state and input constraints given by
\begin{equation}
\label{sec2_cons}
\begin{aligned}
x_{N_i}(t) \in \pazocal{X}_{\pazocal{N}_i} &= \{ x_{N_i} \in \mathbb{R}^{n_{N_i}}: G_i x_{N_i} \leq g_i \}, \\
u_i(t) \in \pazocal{U}_i &= \{ u_i \in \mathbb{R}^{m_i}: H_i u_i \leq h_i \}, \\	
\end{aligned}
\end{equation}
where $G_i \in \mathbb{R}^{q_i \times n_{N_i}}$, $H_i \in \mathbb{R}^{r_i \times m_i}$, $g_i \in \mathbb{R}^{q_i}$ and $h_i \in \mathbb{R}^{r_i}$. Note that $x_i$, $x_{N_i}$ and $u_i$ are decision variables for the $i^{\text{th}}$ subsystem. Without loss of generality, the subsystems are assumed to be coupled only through the states; in the presence of coupled inputs, auxiliary variables can be introduced to ensure this assumption is met as in \cite{darivianakis2019distributed}.

For the $i^{\text{th}}$ subsystem, we consider the local cost function
$J_i \coloneqq \sum_{t=0}^{T-1} \left(
\|x_{N_i}(t)-x_{e_{N_i}}\|^2_{Q_i} + \right.$ $ \left. \|u_i(t)-u_{e_i}\|^2_{R_i} \right)
+ \|x_i(T)-x_{e_i}\|^2_{P_i} + \|x_{e_i}-x_{r_i}\|^2_{S_i}$
where $Q_i \in \mathbb{S}_{++}^{n_{N_i}}$, $R_i \in \mathbb{S}_{++}^{m_i}$, $P_i  \in \mathbb{S}_{++}^{n_i}$, $S_i \in \mathbb{S}_{++}^{n_i}$, $\mathbb{S}_{++}^n$ refers to the set of $n$-by-$n$ symmetric positive definite matrices, $T$ is the prediction horizon, $x_{e_i}$ and $u_{e_i}$ are decision variables\footnote{We use the term ``decision variable’’ to indicate quantities determined by the online optimization problem solved once the initial state has been measured. This is in contrast to quantities like $P_i$, $Q_i$, and $K_i$ (below) that are offline design choices.} defining an equilibrium point for the $i^{\text{th}}$ subsystem, $x_{e_{{N}_i}}$ is a decision variable including the states corresponding to the equilibrium points of the subsystems in the set $\pazocal{N}_i$ and $x_{r_i}$ is the reference point of the $i^{\text{th}}$ subsystem.
We assume that the equilibrium point of the $i^{\text{th}}$ subsystem satisfies 
\begin{equation}
\label{sec2_eq}
	x_{e_i} = A_i x_{e_{N_i}} +B_i u_{e_i}, \quad u_{e_i} = K_i x_{e_{N_i}} + d_i.
\end{equation}
where $K_i \in \mathbb{R}^{m_i \times n_{N_i}}$ is a stabilizing control gain precomputed offline and $d_i$ is a decision variable. To guarantee recursive feasibility, we require that 
\begin{equation}
	\label{sec2_ter}
	x_i(T) \in \pazocal{X}_{f_i} = \{x_i \in \mathbb{R}^{n_i} : (x_i-c_i)^\top P_i (x_i-c_i) \leq \alpha_i^2\}
\end{equation}
where $\pazocal{X}_{f_i}$ is positively invariant under the stabilizing terminal controller $\kappa_i(x_{N_i}) = K_i x_{N_i} + d_i$. 

The matrix $P_i$ of the local terminal cost as well as the terminal control gain matrix $K_i$ can be computed offline by solving an SDP. The matrix $P_i$ is designed in such a way that $V \colon x \mapsto x^\top P x$  with $P=\operatorname{diag}(P_1,...,P_M)$ can be used as a Lyapunov function for the plant controlled by MPC. Following \cite{conte2012distributed,conte2016distributed}, we solve offline the SDP
\begin{equation}
\label{sec2_off}
\begin{aligned}
\max_{ E_i, Y_i, H_i, S_i} \sum_{i=1}^{M} \operatorname{trace} (E_i)
\ \text{s.t.} \
\left\{
\begin{aligned}
    & \forall i \in \{1,...,M\} \\
	& E_i\geq\epsilon_i, \ \eqref{sec2_LMI}, \ H_i \leq S_i, \\
	& \sum_{j \in \pazocal{N}_i} W_{ij} S_j W_{ij}^\top\leq 0,
\end{aligned}
\right.
\end{aligned}
\end{equation}
\noindent 
where $\operatorname{trace}(\cdot)$ is the trace operator of a matrix, $\epsilon_i$ are arbitrarily small positive constants, $E_i=P_i^{-1}$, $Y_{i}=K_iE_{N_i}$, $H_{i}$ and $S_{i}$ are decision variables, $E_{N_i}=\sum_{j \in \pazocal{N}_i} W_{ij}^\top E_j W_{ij}$, $S_{i}$ is assumed to be block diagonal and
\begin{equation}
	\label{sec2_LMI}
	\begin{bmatrix}
	W_i U_i^\top E_i U_i W_i^\top + H_{i} & * & * & * \\
	A_i E_{\pazocal{N}_i} + B_i Y_{i} & E_i & * & * \\
	Q_i^{1/2} E_{i} & 0 & I_{N_i} & * \\
	R_i^{1/2}  Y_{i} & 0 & 0 & I_{m_i}
	\end{bmatrix} \geq 0.
\end{equation}


Denoting the current state of the $i^{\text{th}}$ subsystem by $x_i^\text{init}$, we aim to solve in a distributed way the online OCP
\begin{equation}
\label{sec2_ocp}
\begin{aligned}
 \min_{ \left\{ \begin{aligned}  x_i(t),u_i&(t), d_i \\  x_{e_i},u_{e_i}, &\alpha_i, c_i \end{aligned} \right\}} \ \sum_{i=1}^M J_i
\quad \text{s.t.} \left\{		
\begin{aligned}
& x_i(0)=x_i^\text{init} \\
& \eqref{sec2_dyn},\eqref{sec2_cons}, \eqref{sec2_eq},\eqref{sec2_ter}, \\
& \forall t \in \{0,...,T-1\}, \\
& \forall i \in \{1,...,M\}.
\end{aligned}
\right.
\end{aligned}
\end{equation}
In addition to the decision variables $x_i(t)$ and $u_i(t)$ found in standard distributed MPC \cite{conte2012distributed,conte2016distributed} and the decision variables $x_{e_i}$ and $u_{e_i}$ used in tracking MPC \cite{ferramosca2013cooperative}, we also treat $\alpha_i$ and $c_i$ (parameterizing the terminal set) and $d_i$ (the affine term in the terminal controller) as decision variables to be determined online. Unlike \cite{aboudonia2021distributed}, $K_i$ (the terminal controller gain) is no longer a decision variable.

\section{Distributed MPC Scheme}\label{sec:mpc}

Although the variable $d_i$ does not affect the stability of the terminal dynamics, additional constraints still need to be imposed in \eqref{sec2_ocp} to ensure the positive invariance of the terminal set $\pazocal{X}_{f_i}$, which is parametrized by $\alpha_i$ and $c_i$, under the terminal controller $\kappa_i(x_{N_i})=K_i x_{N_i} + d_i$. We use the following proposition to derive these additional constraints.
\begin{prop}[\hspace{-0.25pt}\cite{darivianakis2019distributed}]
	\label{sec3_prop}
	Each local terminal set $\pazocal{X}_{f_i}$ is positively invariant under the action of the distributed controller $\kappa_i(x_{N_i})$ if for each subsystem $i \in \{1,...,M\}$ and for all $ x_{N_i} \in \Cross_{j \in \pazocal{N}_i} \pazocal{X}_{f_j}$,
	\begin{subequations}
		\label{sec3_prop0}
		\begin{gather}
		\label{sec3_prop1}
		A_ix_{N_i}+B_i\kappa_i(x_{N_i}) \in \pazocal{X}_{f_i}, \\
		\label{sec3_prop2}
		x_{N_i} \in \pazocal{X}_{\pazocal{N}_i}, \\
		\label{sec3_prop3}
		\kappa_i(x_{N_i}) \in \pazocal{U}_{i}.
		\end{gather}
	\end{subequations} 
\end{prop} 

In the sequel, we derive the constraints corresponding to conditions \eqref{sec3_prop1}, \eqref{sec3_prop2} and \eqref{sec3_prop3} in Propositions \ref{sec3_propA}, \ref{sec3_propB} and \ref{sec3_propC}, respectively. For this purpose, we define $\alpha_{N_i}= \sum_{j \in \pazocal{N}_i} \alpha_j W_{ij}^\top W_{ij}$ and $c_{N_i}=\sum_{j \in \pazocal{N}_i} W_{ij}^\top c_j$. 

\begin{prop}
	\label{sec3_propA}
	Condition \eqref{sec3_prop1} holds for all $j \in \pazocal{N}_i$, $x_j \in \pazocal{X}_{f_j}$ if there exist scalars $\lambda_{ij}\geq 0$ such that the linear matrix inequality \eqref{sec3_LMI1} holds where $P_{ij} = W_{ij}^\top P_j W_{ij}^\top$.
	\begin{equation}
	\label{sec3_LMI1}
	\begin{bmatrix}
	P_i^{-1} \alpha_i & (A_i +B_i K_i) \alpha_{{N}_i} &	(A_i+B_iK_i) c_{{N}_i} + B_i d_i - c_i \\
	* & \sum_{j \in \pazocal{N}_i} \lambda_{ij} P_{ij} & 0 \\
	* & * & \alpha_i - \sum_{j \in \pazocal{N}_i} \lambda_{ij}
	\end{bmatrix}
	\geq 0
\end{equation}
\end{prop}
\begin{proof}
    The proof follows that of Proposition III.1 in \cite{aboudonia2021distributed} by making use of the S-lemma and the Schur complement. The only difference is that the resulting matrix inequality \eqref{sec3_LMI1} is linear since $K_i$ is no longer a decision variable.
\end{proof}



\begin{prop}
	\label{sec3_propB}
	Let the $k^{\text{th}}$ row of the matrix $G_i$ be denoted by $G_i^k$ and the $k^{\text{th}}$ element of the matrix $g_i$ by $g_i^k$. Condition \eqref{sec3_prop2} holds for all $j \in \pazocal{N}_i$, $x_j \in \pazocal{X}_{f_j}$ if and only if
	\begin{equation}
	    \label{sec3_LMI2}
	    G_i^k c_{N_i} + \sum_{j \in \pazocal{N}_i} \| G_i^k W_{ij}^\top P_j^{-1/2} \|_2 \alpha_j \leq g_i^k, \quad \forall k \in \{1,...,q_i\}.
	\end{equation}
\end{prop}

\begin{proof}
    First, we define $s_j \in \mathbb{R}^{n_j}$ such that $x_j = c_j + P_j^{-1/2} s_j$ for all $j \in \pazocal{N}_i$ and hence, $x_{N_i} = c_{N_i} + P_{N_i}^{-1/2} s_{N_i}$ where $P_{N_i} = \sum_{j \in \pazocal{N}_i} W_{ij}^\top P_j W_{ij}$ and $s_{N_i} = \sum_{j \in \pazocal{N}_i} W_{ij}^\top s_j $. Thus, condition \eqref{sec3_prop2} can be written as 
    $G_i c_{N_i} + G_i P_{N_i}^{-1/2} s_{N_i} \leq g_i$ for all $j \in \pazocal{N}_i$, $s_j^\top s_j \leq  \alpha_j^2$. This is equivalent to satisfying $G_i^k c_{N_i} + G_i^k P_{N_i}^{-1/2} s_{N_i} \leq g_i^k$ for all $k \in \{1,...,q_i\}$, $j \in \pazocal{N}_i$, $s_j^\top s_j \leq  \alpha_j^2$. By making use of the definitions of $P_{N_i}$ and $s_{N_i}$, we reach that $G_i^k c_{N_i} + \sum_{j \in \pazocal{N}_i} G_i^k W_{ij}^\top P_j^{-1/2} s_{j} \leq g_i^k$ for all $k \in \{1,...,q_i\}$, $j \in \pazocal{N}_i$, $s_j^\top s_j \leq  \alpha_j^2$. Following \cite{banjac2020improving,boyd2004convex}, this robust constraint is satisfied if and only if \eqref{sec3_LMI2} holds.
\end{proof}

\begin{prop}
	\label{sec3_propC}
	Let the $k^{\text{th}}$ row of the matrix $H_i$ be denoted by $H_i^k$ and the $k^{\text{th}}$ element of the matrix $h_i$ by $h_i^k$. Condition \eqref{sec3_prop3} holds for all $j \in \pazocal{N}_i$, $x_j \in \pazocal{X}_{f_j}$ if and only if
	\begin{equation}
	    \label{sec3_LMI3}
	    \begin{aligned}
	        \forall k \in  \{1,&...,r_i\} , \\
	        & H_i^k K_i c_{N_i} + H_i^k d_i + \sum_{j \in \pazocal{N}_i} \| H_i^k K_i W_{ij}^\top P_j^{-1/2} \|_2 \alpha_j \leq h_i^k.
	    \end{aligned}
	\end{equation}
\end{prop}

\begin{proof}
    The proof follows that of Proposition \ref{sec3_propB} by replacing $G_i$ with $H_i K_i$ and $g_i$ with $h_i-H_id_i$.
\end{proof}

By considering the constraints \eqref{sec3_LMI1}, \eqref{sec3_LMI2} and \eqref{sec3_LMI3} for all $i \in \{1,...,M\}$ in the online OCP, we guarantee that the terminal set $X_{f_i}$ is positively invariant under the terminal controller $\kappa_i(x_{N_i})$ for all $i \in \{1,...,M\}$. Thus, the online OCP \eqref{sec2_ocp} then approximates to 
\begin{equation}
\label{sec3_ocp}
\begin{aligned}
 \min_{ \left\{ \begin{aligned}  x_i(t&), u_i(t) \\  x_{e_i},&u_{e_i},\lambda_{ij} \\ \alpha_i&, c_i, d_i  \end{aligned} \right\}} \ \sum_{i=1}^M J_i
\quad \text{s.t.} \left\{		
\begin{aligned}
& x_i(0)=x_i^\text{init} \\
& \eqref{sec2_dyn},\eqref{sec2_cons}, \eqref{sec2_eq},\eqref{sec2_ter} \\
& \eqref{sec3_LMI1}, \eqref{sec3_LMI2}, \eqref{sec3_LMI3}, \lambda_{ij} \geq 0, \\
& \forall t \in \{0,...,T\}, \\
& \forall i \in \{1,...,M\}, \ \forall j \in \pazocal{N}_i.
\end{aligned}
\right.
\end{aligned}
\end{equation}
To the decision variables of \eqref{sec2_ocp}, we have now added $\lambda_{ij}$ in the LMI \eqref{sec3_LMI1}. Note that the terminal set of the $i^{\text{th}}$ subsystem is not necessarily centered around the equilibrium point of the $i^{\text{th}}$ subsystem. (i.e. $c_i \neq x_{e_i}$). 

The following theorem establishes the recursive feasibility and the convergence of the closed-loop dynamics under the resulting controller to the piecewise constant reference $x_{r_i}$.
\begin{theorem}
	\label{th_stability}
	The MPC scheme \eqref{sec3_ocp} is recursively feasible and the corresponding closed-loop system converges to the piecewise constant reference $x_{r_i}$ given that this reference is admissible and changes finite number of times.
\end{theorem}
\begin{proof}
    The proof follows that of Theorem IV.1 in \cite{aboudonia2021distributed} with $K_i$ assumed fixed and no longer a decision variable.
\end{proof}

Note that the feasible region of either the novel scheme \eqref{sec3_ocp} or the one developed in \cite{aboudonia2021distributed} does not coincide in that of the other. This is because the scheme \eqref{sec3_ocp} does not constrain the center of the terminal set to the equilibrium point defined by $x_{e_i}$, but uses a predefined terminal controller computed offline. On the contrary, the scheme in \cite{aboudonia2021distributed} computes the terminal controller online, but constrains the center of the terminal set to the equilibrium point. Both schemes, however, have larger feasible sets than other schemes porposed in the literature. We show in Theorem \ref{th_feasibility} that the feasible region of the developed scheme \eqref{sec3_ocp} contains that of the scheme proposed in \cite{aboudonia2020distributed} (see \eqref{app_ocp} in Appendix).

\begin{theorem}
	\label{th_feasibility} 
	For a given initial condition, the distributed MPC problem \eqref{sec3_ocp} is feasible if the distributed MPC problem in \eqref{app_ocp} is feasible.
\end{theorem}

\begin{proof}
    For a given initial condition $x_i^\text{init}, \ i \in \{1,...,M\}$, assume that $x_i(t), \ u_i(t), \ \alpha_i, \ c_i, \ \rho_{ij}, \ \sigma_{ij}^k, \ \tau_{ij}^l$ for all $t \in \{0,...,T\}, \ i \in \{1,...,M\}, \ j \in \pazocal{N}_i, \ k \in \{1,...,n_{q_i}\}, \ l \in \{1,...,n_{r_i}\}$ is a feasible solution for the MPC problem \eqref{app_ocp}. This solution then satisfies the constraints in \eqref{app_ocp}. The goal is to prove that it also satisfies the constraints in \eqref{sec3_ocp}. Note that this solution satisfies \eqref{sec2_dyn}, \eqref{sec2_cons}, \eqref{sec2_ter} and ${x}_i(0)={x}_{i_0}$, for all $ t \in \{0,...,T\}$ and $ \ i \in \{1,...,M\}$ in (\ref{sec3_ocp}) since the same constraints occur in \eqref{app_ocp}.
	By comparing the cost functions of both schemes, we deduce that the reference point $x_{r_i}$ and the equilibrium point given by $x_{e_i}$ and $u_{e_i}$ are all equal to zero for all $i \in \{1,...,M\}$ and hence, \eqref{sec2_eq} is satisfied in \eqref{sec3_ocp} for $d_i=0$. Therefore, \eqref{sec3_LMI1} is satisfied since this constraint is the same as \eqref{app_LMI1} for $d_i=0$ and $\lambda_{ij}={\rho}_{ij}$.
	Finally, note that \eqref{sec3_LMI2} and \eqref{sec3_LMI3} are equivalent to \eqref{sec3_prop2} and \eqref{sec3_prop3}. On the other side, \eqref{app_LMI2} and \eqref{app_LMI3} are sufficient conditions for \eqref{sec3_prop2} and \eqref{sec3_prop3} assuming that $d_i=0$. Therefore, \eqref{app_LMI2} and \eqref{app_LMI3} imply \eqref{sec3_LMI2} and \eqref{sec3_LMI3} for $d_i=0$.
\end{proof}

The online OCP \eqref{sec3_ocp} of the proposed scheme is formulated as an SDP due to the LMI \eqref{sec3_LMI1}. To reduce the computational effort, the SDP can be approximated by an SOCP using diagonal dominance \cite{majumdar2020recent,ahmadi2019dsos}.
\begin{definition}[\hspace{-0.01cm}\cite{ahmadi2019dsos}]
    \label{sec3_def}
    A symmetric matrix $A$ is diagonally dominant if $a_{ii} \geq \sum_{j \neq i} |a_{ij}|$ for all $i$. 
\end{definition}
The following proposition can be used to approximate the LMI \eqref{sec3_LMI1} by a set of linear inequalities.
\begin{prop}[\hspace{-0.01cm}\cite{berman2003completely}]
    \label{sec3_prop_dd}
    A diagonally dominant symmetric matrix is positive semidefinite.
\end{prop}

Based on Definition \ref{sec3_def} and Proposition \ref{sec3_prop}, LMI \eqref{sec3_LMI1} can be approximated using the linear inequalities \eqref{sec3_dd} (found overleaf in single column) where $\{\cdot\}_{jk}$ is the element in the $j^{\text{th}}$ row and $k^{\text{th}}$ column of a matrix and $|\cdot|$ is a matrix with the absolute values of the elements in the original matrix.

    \begin{table*}
    \normalsize
    \begin{subequations}
    \label{sec3_dd}
    \begin{gather}
        \left\{ P_i^{-1} \alpha_i \right\}_{kk} 
        \geq 
        \sum\nolimits_{l=1}^{n_i} \left\{ |P_i^{-1}| \alpha_i \right\}_{kl}
        -
        \left\{ |P_i^{-1}| \alpha_i \right\}_{kk}
        +
        \sum\nolimits_{l=1}^{n_{N_i}} \left\{ |A_i + B_i K_i| \alpha_{N_i} \right\}_{kl}
        +
        \left\{ b_i \right\}_{k},
        \ \forall k \in \{1,...,n_i\}
    \\
    \vspace{0.25cm}
    \begin{split}
        \left\{ \sum\nolimits_{j \in \pazocal{N}_i} \lambda_{ij} P_{ij} \right\}_{kk} 
        \geq 
        \sum\nolimits_{l=1}^{n_{N_i}} \left\{ \sum\nolimits_{j \in \pazocal{N}_i} \lambda_{ij} |P_{ij}| \right\}_{kl}
        -
        \left\{ \sum\nolimits_{j \in \pazocal{N}_i} \lambda_{ij} |P_{ij}| \right\}_{kk}
        +
        \sum\nolimits_{l=1}^{n_i} \left\{ \left(|A_i + B_i K_i| \alpha_i\right)^\top \right\}_{kl}, \\
        \ \forall k \in \{1,...,n_{N_i}\}
    \end{split}
    \\
        \alpha_i - \sum\nolimits_{j \in \pazocal{N}_i} \lambda_{ij}
        \geq 
        \sum\nolimits_{l=1}^{n_i} \left\{ b_i \right\}_{l},
    \quad \quad \quad
    -b_i \leq (A_i+B_iK_i) c_{N_i} + B_i d_i - c_i \leq b_i
    \end{gather}
    \rule{\textwidth}{0.4pt}
    \end{subequations}
    \vspace{-0.6cm}
    \end{table*}

In this case, the online OCP \eqref{sec3_ocp} is approximated by    
\begin{equation}
\label{sec3_ocpdd}
\begin{aligned}
 \min_{ \left\{ \begin{aligned}  x_i(t&), u_i(t) \\  x_{e_i},&u_{e_i},\lambda_{ij} \\ \alpha_i, &c_i, d_i, b_i  \end{aligned} \right\}} \ \sum_{i=1}^M J_i
\quad \text{s.t.} \left\{		
\begin{aligned}
& x_i(0)=x_i^\text{init}, \\
& \eqref{sec2_dyn},\eqref{sec2_cons}, \eqref{sec2_eq},\eqref{sec2_ter}, \\
& \eqref{sec3_dd}, \eqref{sec3_LMI2}, \eqref{sec3_LMI3}, \lambda_{ij} \geq 0, \\
& \forall t \in \{0,...,T\}, \\
& \forall i \in \{1,...,M\}, \ \forall j \in \pazocal{N}_i.
\end{aligned}
\right.
\end{aligned}
\end{equation}
It is easy to verify that Theorem \ref{th_stability} still holds for \eqref{sec3_ocpdd} which can be cast as an SOCP and not a QP due to constraint \eqref{sec2_ter}.

\section{Distributed Implementation}\label{sec:implementation}

\label{ADMM}

The online OCPs \eqref{sec3_ocp} and \eqref{sec3_ocpdd} can be solved using one of the many distributed optimization techniques proposed in the literature \cite{nedic2018distributed}. Some of these techniques do not require a central coordinator such as the distributed primal-dual algorithm (DPDA) \cite{aybat2019distributed} and some variants of ADMM \cite{boyd2011distributed}; the numerical comparison in \cite{banjac2019decentralized} suggests that ADMM outperforms DPDA for certain classes of problems. Motivated by this, we implement \eqref{sec3_ocp} and \eqref{sec3_ocpdd} using ADMM \cite[Section 7]{boyd2011distributed}, where each subsystem solves a local optimization problem iteratively while sharing information with its neighbours. Note that two neighbours $i$ and $j$ share the variables $x_i(t)$, $x_j(t)$, $x_{e_i}$, $x_{e_j}$, $\alpha_i$, $\alpha_j$, $c_i$ and $c_j$ for all $t \in \{1,...,T\}$. Thus, the shared variables of the $i^{\text{th}}$ subsystem are $w_{N_i}=(x_{N_i}(t)|_{t=\{0,...,T\}},x_{e_{N_i}},\operatorname{diag}(\alpha_{N_i}),c_{N_i})$, whereas its non-shared variables are $v_i=(u_i(t)|_{t=\{0,...,T-1\}},$ $u_{e_i},d_i,\lambda_{ij}|_{j \in \pazocal{N}_i})$; to simplify notation we also define $w_i=(x_i(t)|_{t=\{0,...,T\}},x_{e_i},\alpha_i,c_i)$. In the sequel, we show briefly how to solve \eqref{sec3_ocp} and \eqref{sec3_ocpdd} using ADMM; see \cite[Section 7]{boyd2011distributed} for more details.

First, we define for each subsystem the local augmented cost function $\tilde{J}_i$, which encodes both its local cost function $J_i$ and its constraints through indicator functions. Hence, the online OCP \eqref{sec3_ocp} and \eqref{sec3_ocpdd} are given by
\begin{equation}
    \label{sec3_admm1}
    \min_{w_{N_i},v_i} \sum\nolimits_{i=1}^{M} \tilde{J}_i(w_{N_i},v_i).
\end{equation}
We then define a local copy for each shared decision variable in the augmented cost function of the $i^{\text{th}}$ subsystem and denote it by $(\cdot)^{(i)}$. To obtain a feasible solution, the local copies of the same decision variable existing in different augmented cost functions should be equal. Hence, we define a global copy $z$ comprising all shared decision variables. We refer to the subvector of $z$ corresponding to $w_{N_i}$ as $z_{N_i}$ and to the subvector of $z$ corresponding to $w_i$ as $z_i$. For two neighbours $i$ and $j$, the vectors $z_{N_i}$ and $z_{N_j}$ overlap. The vectors $z_i$ (a subvector of $z_{N_i}$) and $z_j$ (a subvector of $z_{N_j}$) do not overlap however. Hence, the online OCP \eqref{sec3_admm1} becomes
\[
    \min_{w_{N_i}^{(i)},v_i} \sum\nolimits_{i=1}^{M} \tilde{J}_i(w_{N_i}^{(i)},v_i) \ s.t. \ w_{N_i}^{(i)} = z_{N_i} \quad \forall \ i \in \{1,...,M\},
\]
and its augmented Lagrangian is given by
\[
    L = \sum\nolimits_{i=1}^{M} L_i (w_{N_i}^{(i)},v_i,z_{N_i},y_{N_i}),
\]
where $y_{N_i}$ is the Lagrange multiplier computed by the $i^{\text{th}}$ subsystem and
\[
    L_i(w_{N_i}^{(i)},v_i,z_{N_i},y_{N_i}) = \tilde{J}_i(w_{N_i}^{(i)},v_i) + \tfrac{\rho}{2} \|w_{N_i}^{(i)}-z_{N_i}+\tfrac{1}{\rho}y_{N_i}\|_2^2.
\]






\begin{algorithm}[t]
	\caption{Distributed MPC for Reference Tracking}
	\label{alg}
	\begin{algorithmic}[1]
		\INPUT System matrices $A_i$, $B_i$, constraint sets $\pazocal{X}_{N_i}$, $\pazocal{U}_i$, cost function matrices $Q_i$, $R_i$, $S_i$ map $W_{ij}$, initial condition $x_{0_i}$, prediction horizon $T$, ADMM step size $\rho$, ADMM maximum time $T_{max}$
		\OUTPUT $x_{N_i}(t)$, $u_i(t)$, $\forall t \in \{0,...,T\}$, $\alpha_i$, $c_i$, $d_i$, $\lambda_{ij}$, $ \forall j \in \pazocal{N}_i$
		\STATE Solve \eqref{sec2_off} offline 
		\STATE Set $z_{N_i}[0]=0$ and $y_{N_i}[0]=0$
		\WHILE{true}
		\STATE Set $k=1$
		\STATE \textbf{do}
		\STATE $\ \ $ Set $z_{N_i}^\text{prev} = z_{N_i}[k-1]$
		and $y_{N_i}^\text{prev} = y_{N_i}[k-1]$
	    \STATE $\ \ $ $(w_{N_i}^{(i)}[k],v_i[k]) = \argmin\limits_{w_{N_i}^{(i)},v_i}
        L_i\left(w_{N_i}^{(i)},v_i,z_{N_i}^\text{prev}, y_{N_i}^\text{prev}\right)$
	    \STATE $\ \ $ Share $w_j^{(i)}$ and $w_i^{(j)}$ with $j \in \pazocal{N}_i$
	    \STATE $\ \ $ $z_i[k] = \frac{1}{|\pazocal{N}_i|} \sum_{j \in \pazocal{N}_i} w_i^{(j)}[k]$
	    \STATE $\ \ $ Share $z_i$ with all $j \in \pazocal{N}_i$
		\STATE $\ \ $ $y_{N_i}[k] = y_{N_i}[k-1] + \rho \left( w_{N_i}^{(i)}[k] - z_{N_i}^{(i)}[k] \right)$
		\STATE $\ \ $ Set $k=k+1$
		\STATE \textbf{until time limit $T_{max}$ is reached}
		\STATE Apply the first control input  $u_i(0)$ to the plant (\ref{sec2_dyn})
		\STATE Measure/estimate the new state $x_i^\text{init}$ 
		\ENDWHILE
	\end{algorithmic}
\end{algorithm}


By using the augmented Lagrangian, we can run the iterative ADMM algorithm where each subsystem performs three steps in each iteration. First, each subsystem solves a local optimization problem to update its local variables. Then, each subsystem updates a subvector of the global copy using local information only. Finally, each subsystem updates its Lagrange multipliers. At each timestep, the ADMM algorithm terminates after a predefined time determined by the sampling time available for computations. Algorithm \ref{alg} shows how to implement the proposed scheme using ADMM where $T[k]$ is the time required to perform one iteration.

\section{Simulations}\label{ref:simulations}

We evaluate the efficacy of the proposed schemes \eqref{sec3_ocp} and \eqref{sec3_ocpdd} and compare them to \cite{aboudonia2021distributed} in terms of performance and computational cost; the numerical results in \cite{aboudonia2021distributed} suggest that this scheme has a  larger feasible region and tends to show better performance compared to other distributed schemes in the literature. We denote \eqref{sec3_ocp} by DST, \eqref{sec3_ocpdd} by DST+DD and the scheme developed in \cite{aboudonia2021distributed} by RTI. We solve all optimization problems using MATLAB with YALMIP \cite{Lofberg2004yalmip} and MOSEK \cite{mosek} on a computer equipped with 16-GB RAM and a 1.9-GHz Intel core i7-8550U processor.
\begin{figure}
	\centering
	\scalebox{1}{		\begin{tikzpicture}[thick,scale=0.8, every node/.style={scale=0.8}]
			\node[draw,thick,rectangle,rounded corners=0.25cm,minimum size=.8cm] (pga1) {PGA 1};
			\node[draw,thick,rectangle,rounded corners=0.25cm,minimum size=.8cm, right = 0.8cm of pga1] (pga2) {PGA 2};
			\node[draw,thick,rectangle,rounded corners=0.25cm,minimum size=.8cm, right = 0.8cm of pga1, below = 0.6cm of pga2] (pga3) {PGA 3};
			\node[draw,thick,rectangle,rounded corners=0.25cm,minimum size=.8cm, right = 0.8cm of pga2] (pga5) {PGA 5};
			\node[draw,thick,rectangle,rounded corners=0.25cm,minimum size=.8cm, right = 0.8cm of pga2, below = 0.6cm of pga5] (pga4) {PGA 4};
			\node[draw,thick,rectangle,rounded corners=0.25cm,minimum size=.8cm, right = 0.5cm of pga5] (pga6) {PGA 6};
			\node[draw,thick,rectangle,rounded corners=0.25cm,minimum size=.8cm, right = 0.8cm of pga5, below = 0.6cm of pga6] (pga7) {PGA 7};
			\draw[Stealth-Stealth,thick] (pga1) -- (pga2);
			\draw[Stealth-Stealth,thick] (pga2) -- (pga3);
			\draw[Stealth-Stealth,thick] (pga2) -- (pga5);
			\draw[Stealth-Stealth,thick] (pga3) -- (pga4);
			\draw[Stealth-Stealth,thick] (pga4) -- (pga5);
			\draw[Stealth-Stealth,thick] (pga5) -- (pga6);
			\draw[Stealth-Stealth,thick] (pga5) -- (pga7);
		\end{tikzpicture}}
	\caption{Power network topology}
	\label{Topology}
	\vspace{-0.4cm}
\end{figure}
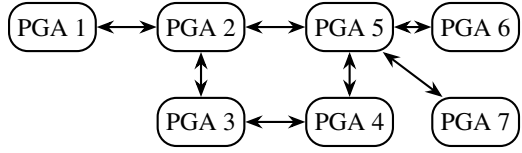
\begin{table}
    \centering
    \caption{Power Network Parameters}
    \renewcommand{\arraystretch}{1.3}
    \begin{tabular}{crlllll}
        \toprule
        PGA & $H_i$ & $D_i$ & $R_{t_i}$ & $T_{t_i}$ & $T_{g_i}$ & $p_i^{max}$ \\
        \midrule
        1 & 12 &   0.05 & 0.7 & 0.65 & 0.1 &  0.5 \\
        2 & 10 & 0.0625 & 0.9 &  0.4 & 0.1 & 0.65 \\
        3 &  8 &    0.8 & 0.9 &  0.3 & 0.1 & 0.65 \\
        4 &  8 &    0.8 & 0.7 &  0.6 & 0.1 & 0.55 \\
        5 &  8 &    0.8 & 0.9 &  0.3 & 0.1 & 0.65 \\
        6 & 10 & 0.0625 & 0.9 &  0.4 & 0.1 & 0.65 \\
        7 & 12 &   0.05 & 0.7 & 0.65 & 0.1 &  0.5 \\
        \bottomrule
    \end{tabular}
    \label{tab:my_label}
    \vspace{-0.6cm}
\end{table}

We use a power network case study \cite{riverso2013plug} comprising a set of power generation areas (PGAs), each of which represents one subsystem in the network (Fig.\ref{Topology}). 
The state and input vectors of the $i^{\text{th}}$ PGA are $x_i=[\Delta \theta_i, \ \Delta \omega_i, \ \Delta P_{M_i} - \Delta P_{L_i}, \ \Delta P_{V_i} - \Delta P_{L_i}]^\top$ and $u_i = \Delta P_{R_i} - \Delta P_{L_i}$, where $\Delta \theta_i$ represents the angular displacement deviation, $\Delta \omega_i$ the angular velocity deviation, $\Delta P_{M_i}$ the mechanical power deviation, $\Delta P_{V_i}$ the steam valve position deviation, $\Delta P_{L_i}$ the load change deviation and $\Delta P_{R_i}$ the reference set power deviation.
Following \cite{DBLP:journals/corr/abs-1302-0226}, the dynamics of the $i^{\text{th}}$ PGA is given by
$\dot{x}_i=\sum_{j \in \pazocal{N}_i} A_{ij}x_j + B_iu_i$ where $A_{ij}$ with $i=j$ and $B_i$ are 
\begin{equation}
    A_{ii} = 
    \begin{bmatrix}
        0 & 1 & 0 & 0 \\
        -\sum_{j \in \pazocal{N}_i} \frac{P_{ij}}{2H_i} & \frac{-D_i}{2H_i} & \frac{1}{2H_i} & 0 \\
        0 & 0 & \frac{-1}{T_{t_i}} & \frac{1}{T_{t_i}} \\
        0 & \frac{-1}{R_{t_i} T_{g_i}} & 0 & \frac{-1}{T_{g_i}} \\
    \end{bmatrix},
    \ 
    B_i = 
    \begin{bmatrix}
        0 \\ 0 \\ 0 \\ \frac{1}{T_{g_i}}. \\    
    \end{bmatrix}
\end{equation}
The entries of the matrix $A_{ij}$ with $i \neq j \in \pazocal{N}_i$ are all zeros except for the one in the second row and first column which equals $\frac{P_{ij}}{2H_i}$. The PGA parameters $H_i$, $D_i$, $R_{t_i}$, $T_{t_i}$ and $T_{g_i}$ are listed in Table 1 for each PGA. The parameter $P_{ij}$ describes the coupling between the two neighbours $i$ and $j$ where $P_{ij}=P_{ji}$, $P_{12}=4$, $P_{23}=2$, $P_{25}=1$, $P_{34}=2$ $P_{45}=2$ $P_{56}=3$ $P_{56}=3$.
It is easy to verify that the continuous-time dynamics of each PGA has the same structure of the discrete-time dynamics \eqref{sec2_dyn}. To preserve this structure after discretization, we use the Frobenius-norm-based discretization method \cite{souza2015discretisation} with a sampling time of one second. Each PGA is subject to the constraints, $|\Delta \theta_i| \leq 0.1$ and $|\Delta P_{R_i}| \leq p_i^{max}$ where $p_i^{max}$ of all PGAs are listed in Table 1.
The weights of the cost function are given by $R_i=0.1$, $S_i=\operatorname{diag}(1000,1000,10,10)$, $W_{ij} Q_{N_i} W_{ij}^\top = 0.99S_c$ if $i=j$ and $0.01S_c$ if $i \neq j$. The matrix $P_i$ and the controller $K_i$ are computed based on \eqref{sec2_off}. The prediction horizon is given by $T=5$.

First, we compare the closed-loop cost $J_s$ of all schemes for 25 randomly-generated target points $x_{r_i}$ where $J_s= \sum_{i=1}^M \sum_{t=0}^{T_{sim}} \left( \| x_{N_i}^t(0)-x_{r_{N_i}} \|^2_{Q_i} + \| u_i^t(0)-u_{r_i} \|^2_{R_i} \right)$, the subscript $s$ refers to scheme $s \in \{\text{DST, \ DST+DD, \ RTI}\}$ and the superscript $t$ refers to the optimal solution at timestep $t$. We solve the OCP of all considered schemes recursively for $T_{sim}=10$ timesteps and centrally to compare the costs when solved to optimality. It is found that all schemes yield almost the same closed-loop cost (omitted in the interest of space). 
\begin{figure}
    \centering
	\includegraphics[scale=0.25]{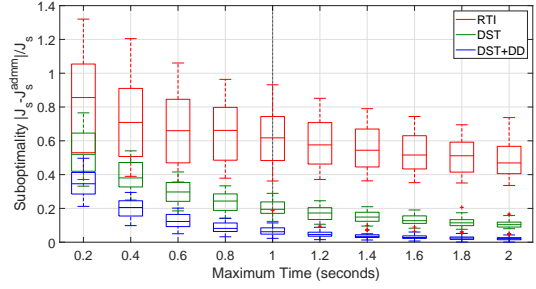}
	\caption{Median, interquartile range $(25\%-75\%)$, minimum, maximum and outliers of the suboptimality $|J_s^{admm}-J_s|/J_s$ vs the maximum computation time per timestep for the three schemes; the dotted black line refers to the used sampling time of 1 second.}
	\label{Fig1}
	\vspace{-0.6cm}
\end{figure}
\begin{figure}
    \centering
	\includegraphics[scale=0.25]{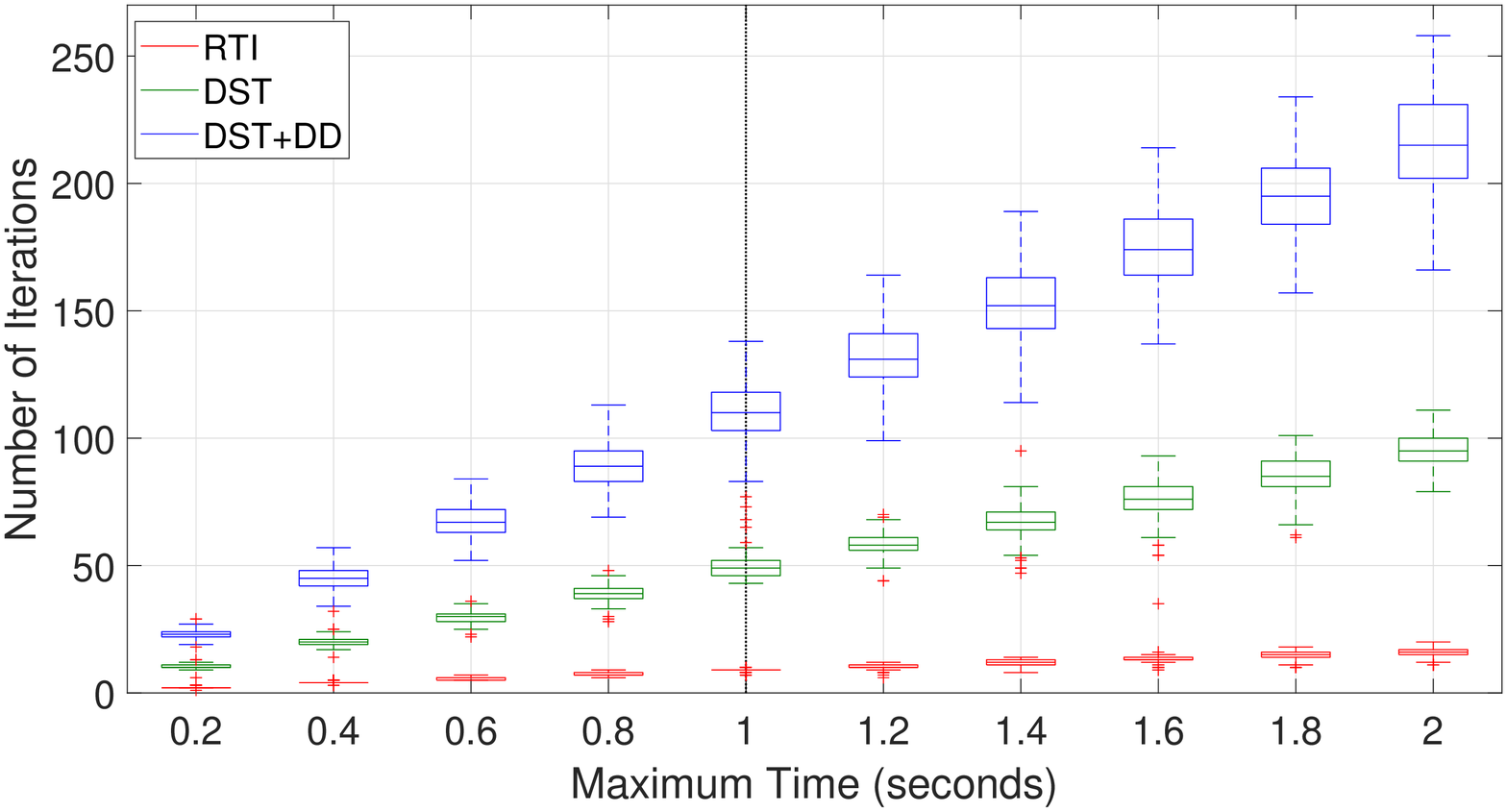}
	\caption{Median, interquartile range $(25\%-75\%)$, minimum, maximum and outliers of the number of ADMM iterations vs the maximum computation time per timestep for the three schemes; the dotted black line refers to the used sampling time of 1 second.}
	\label{Fig2}
	\vspace{-0.6cm}
\end{figure}

Next, we solve the OCP of all schemes in a distributed fashion as described in Section \ref{ADMM}. We denote the resulting closed-loop cost obtained by scheme $s$ by $J_s^\text{admm}$. Note that $J_s^\text{admm}$ converges to $J_s$ only asymptotically in the number of ADMM iterations. Hence, we choose a termination condition based on a pre-defined maximum time $T_\text{max}$ to imitate the amount of computation time available for each sampling time. Fig.\ref{Fig1} shows a boxplot for the suboptimality $|J_s-J_s^\text{admm}|/J_s$ of each scheme $s$ where $T_\text{max} = 0.2r$ seconds and $r \in \{1,...,10\}$ (Equivalently $20\%$ to $200\%$ of the sampling time of 1 second). Note that DST+DD has better convergence properties than DST which, in turn, outperforms RTI. This is mainly because ADMM can perform the highest number of iterations with DST+DD and the lowest number of iterations with DST. This is obvious in Fig.\ref{Fig2} which shows a boxplot for the number of ADMM iterations required within $T_\text{max}$ at each timestep for each initial condition. The number of iterations of  RTI is the lowest possibly due to the larger number of constraints (in particular, LMIs) and decision variables.
On the other hand, the number of iterations of DST+DD is the highest possibly because the resulting OCP is SOCP-representable. Notice that the communication and actuation time is not considered here since all computations are performed in simulations on a single processor. Note, however, that all schemes are using the same ADMM algorithm and communicating the same information. Hence, their communication demand is expected to be almost the same.

\section{Conclusions}\label{conclusion}

A novel distributed MPC scheme with reconfigurable terminal sets is proposed for tracking piecewise constant references. The resulting OCP is amenable to distributed optimization techniques. The effectiveness of the proposed scheme is explored using a power network case study. It is found that the proposed scheme requires less computation, hence yields better performance compared to existing schemes if computation time is limited. Future work includes experimental demonstration of the proposed schemes.




\section*{APPENDIX}

We present here the distributed MPC scheme developed in \cite{aboudonia2020distributed} to compare it to the developed scheme \eqref{sec3_ocp} in Theorem~\ref{th_feasibility}. The online OCP in \cite{aboudonia2020distributed} is given by
\begin{equation}
\label{app_ocp}
\begin{aligned}
 &\min_{ \left\{ \begin{aligned}  x_i(t), u_i(t), \alpha_i \\ c_i, \rho_{ij}, \sigma_{ij}^k, \tau_{ij}^l \end{aligned} \right\} } \ \sum_{i=1}^M \sum_{t=0}^{T-1} \left[
\|x_{N_i}(t)\|^2_{Q_i} + \|u_i(t)\|^2_{R_i} \right]
+ \|x_i(T)\|^2_{P_i}, \\
& \text{s.t.} \left\{		
\begin{aligned}
& \eqref{sec2_dyn},\eqref{sec2_cons}, \eqref{app_ter} - \eqref{app_LMI3}, \ x_i(0)=x_i^\text{init} \\
& \rho_{ij} \geq 0, \ \sigma_{ij}^k \geq 0, \ \tau_{ij}^l \geq 0, \ \forall j \in \pazocal{N}_i \\
& \ \forall k \in \{0,...,q_i\}, \ \forall l \in \{1,...,r_i\} \\
& \ \forall t \in \{0,...,T-1\}, \ \forall i \in \{1,...,M\}, \ \forall j \in \pazocal{N}_i.
\end{aligned}
\right.
\end{aligned}
\end{equation}
\begin{equation}
    \label{app_ter}
    \begin{bmatrix}
        P_i^{-1} \alpha_i & x_i - c_i \\
        * & \alpha_i
    \end{bmatrix}
    \geq 0.
\end{equation}
\begin{equation}
	\label{app_LMI1}
	\begin{bmatrix}
	P_i^{-1} \alpha_i & (A_i + B_i K_i) \alpha_{N_i} &	(A_i + B_i K_i) c_{N_i} - c_i
	\\
	* & \sum_{j \in \pazocal{N}_i} \rho_{ij} P_{ij} & 0 \\
	* & * & \alpha_i - \sum_{j \in \pazocal{N}_i} \rho_{ij}
	\end{bmatrix}
	\geq 0.
\end{equation}
\begin{equation}
	\label{app_LMI2}
	\begin{aligned}
	\begin{bmatrix}
	\sum_{j\in\pazocal{N}_i}\sigma_{ij}^k P_{ij} & \frac{1}{2} \alpha_{N_i} G_{i_i}^{k^\top} \\
	* & g_{i}^k - G_i^k   c_{N_i} - \sum_{j\in\pazocal{N}_i}\sigma_{ij}^k
	\end{bmatrix}
	\geq 0.
	\end{aligned}
\end{equation}
\begin{equation}
	\label{app_LMI3}
	\begin{aligned}
	\begin{bmatrix}
	\sum_{j \in \pazocal{N}_i}\tau_{ij}^l P_{ij} & \frac{1}{2} \alpha_{N_i}
	K_i^\top  H_i^{l^\top} \\
	* & h_i^l - H_i^l K_i c_{N_i} - \sum_{j\in\pazocal{N}_i}\tau_{ij}^l
	\end{bmatrix}
	\geq 0.
	\end{aligned}
\end{equation}
According to \cite{aboudonia2020distributed}, \eqref{app_LMI1}, \eqref{app_LMI2} and \eqref{app_LMI3} are only sufficient conditions for the constraints in Proposition~\ref{sec3_prop}.


\section*{ACKNOWLEDGMENT}

We would like to thank Prof. Roy Smith and Dr. Georgios Darivianakis for the fruitful discussions on the topic.

\bibliographystyle{unsrt}
\bibliography{references}

\end{document}